\documentclass[aps,pra,reprint,superscriptaddress,longbibliography]{revtex4-2}

\usepackage{iftex}
\ifXeTeX
  \usepackage{fontspec}
  \defaultfontfeatures{Ligatures=TeX}
\else
  \usepackage[T1]{fontenc}
  \usepackage[utf8]{inputenc}
\fi

\usepackage{amsmath,amssymb,amsthm,mathtools,bm}
\usepackage{graphicx}
\usepackage{xcolor}
\usepackage{microtype}
\usepackage[colorlinks=true,linkcolor=blue,citecolor=blue,urlcolor=blue]{hyperref}

\newtheorem{theorem}{Theorem}
\newtheorem{proposition}{Proposition}
\newtheorem{lemma}{Lemma}
\newtheorem{corollary}{Corollary}
\theoremstyle{definition}
\newtheorem{assumption}{Assumption}
\theoremstyle{remark}

\newcommand{\R}{\mathbb{R}}
\newcommand{\C}{\mathbb{C}}
\newcommand{\I}{\mathbb{I}}
\newcommand{\SU}{\mathrm{SU}}
\newcommand{\SOg}{\mathrm{SO}}
\newcommand{\dd}{\mathrm{d}}
\newcommand{\e}{\mathrm{e}}
\newcommand{\ii}{\mathrm{i}}

\newcommand{\spec}{\operatorname{spec}}

\newcommand{\mcO}{\mathcal{O}}
\newcommand{\mcD}{\mathfrak{D}}
\newcommand{\mcH}{\mathcal{H}}
\newcommand{\mcK}{\mathcal{K}}

\newcommand{\mcOq}{\widetilde{\mathcal{O}}_A}
\newcommand{\vect}[1]{\bm{#1}}
\newcommand{\norm}[1]{\left\lVert #1 \right\rVert}
\newcommand{\abs}[1]{\left\lvert #1 \right\rvert}

\begin{document}

\title{\texorpdfstring{Galilean One-Particle Kinematics \\from a Smooth Family of Reference States}{Galilean One-Particle Kinematics from a Smooth Family of Reference States}}

\author{Jianshuo Gao}
\affiliation{Peking University, School of Physics}
\date{\today}

\begin{abstract}
Giannelli and Chiribella derived an observable-generator duality for energy from a collision model of informational nonequilibrium. We study a continuous-variable version aimed at the Galilean one-particle sector. A smooth family of reference states around an isotropic equilibrium supplies time, translation, rotation, and boost directions. The local observable-generator correspondence is obtained by differentiating a smooth extension of the single-state duality map, and the norm-one property of localization is obtained from a fiducial focusing assumption together with covariance. Combined with the standard smearing form of covariant localization observables, this yields sharp localization. With local inertial composition, the spin-cover action of rotations, and a central boost-translation holonomy, every irreducible sector is unitarily equivalent to $L^2(\R^3)\otimes\C^{2s+1}$. In that representation translations are generated by the canonical momentum, the holonomy is a scalar mass $m>0$, boosts are generated at $t=0$ by $m\vect{X}$, the Hamiltonian is $\vect{P}^{\,2}/(2m)+E_0$, and $\vect{J}=\vect{X}\times\vect{P}+\vect{S}$.
\end{abstract}

\maketitle

\section{Introduction}

The information-theoretic reconstruction program has been especially successful at explaining state spaces and measurement structures, but it has been less explicit about the origin of dynamical observables. A major recent step was the derivation by Giannelli and Chiribella of an observable-generator duality for energy from a collision model in which reversible dynamics arises from informational nonequilibrium between a target system and a stream of identically prepared ancillas \cite{giannelli2026information}. In their framework, a reference state $\sigma$ determines a reversible one-parameter dynamics and, under strong self-duality and uniqueness of the invariant state, also determines the corresponding energy observable. The construction is finite dimensional, compact, and tailored to time evolution.

Canonical momentum, angular momentum, and mass are subtler. The relevant symmetries are noncompact, the corresponding observables have continuous spectra, localization is described operationally by covariant positive-operator-valued measures (POVMs), and Galilei boosts are needed to upgrade a bare translation generator into the mechanical momentum $m\dot{\vect{X}}$. In addition, mass enters nonrelativistic one-particle kinematics not as an ordinary Noether charge but as the central extension parameter of the Galilei group \cite{bargmann1954unitary,levyleblond1963galilei,grigore1994projective,figueroaofarrill2025galilei}. A derivation of momentum and angular momentum therefore has to explain not only which generators appear, but also why localization is sharp, why the momentum observable is canonical, and why the boost-translation commutator is controlled by a scalar mass.

Relational and quantum-reference-frame approaches have emphasized that frame changes should be treated operationally and, ultimately, physically \cite{loveridge2018symmetry,giacomini2019quantum,delahamette2020quantum}. Our goal is more specific. We do not attempt a general theory of quantum reference-frame transformations. Instead, we ask which irreducible one-particle kinematics is forced once infinitesimal frame directions arise from informational nonequilibrium and are combined with covariant localization and a minimal holonomy principle.

Here we replace the single reference state of Ref.~\cite{giannelli2026information} with a smooth manifold of reference states around an isotropic equilibrium. Its tangent space at equilibrium carries one temporal direction, three translational directions, three rotational directions, and three boost directions. Linearizing the collisional dynamics along this manifold gives a Lie-algebra-valued connection. The directional observable-generator map then comes from differentiating the single-state duality locally, while the norm-one property of localization is recovered from a fiducial focusing condition and covariance. Combined with the standard smearing structure of covariant localization observables, this forces sharp localization.

Under these assumptions one recovers the usual irreducible Galilean one-particle sector. The translation observable agrees with the canonical momentum selected by sharp localization together with boost covariance of momentum, angular momentum takes the orbital-plus-spin form, and mass appears through the holonomy of boost-translation loops.

The paper is intentionally narrow. We do \emph{not} try to reconstruct infinite-dimensional quantum theory from the finite-dimensional axioms of Ref.~\cite{giannelli2026information}. The result depends on a smooth local completion of the single-state duality, a covariant localization observable with fiducial focusing and the standard smearing structure, local inertial composition, and a central holonomy condition.

\section{Reference-state bundle and local generators}

Consider a target system $A$ and an auxiliary system $F$ that plays the role of a reference frame. In Ref.~\cite{giannelli2026information}, each deterministic state $\sigma\in \mathrm{St}_1(F)$ determines, through a stationary collision model, a reversible dynamics $U_{t,\sigma}=\exp[t G_\sigma]$ of the target system. To describe inertial frame changes, a single control state is not enough; one needs several independent directions corresponding to translations, rotations, boosts, and time. We therefore introduce a smooth finite-dimensional manifold $\Sigma\subset \mathrm{St}_1(F)$ of reference states containing an isotropic equilibrium state $\chi$.

For any smooth curve $\varepsilon\mapsto \sigma_\varepsilon\in\Sigma$ with $\sigma_0=\chi$ and tangent vector $v=\dot\sigma_0\in T_\chi\Sigma$, the linearization of the collisional generator defines
\begin{equation}
\omega(v):=\left.\frac{\dd}{\dd \varepsilon}\right\rvert_{\varepsilon=0} G_{\sigma_\varepsilon}.
\label{eq:omega}
\end{equation}
We interpret $\omega$ as a connection one-form on the reference-state bundle at equilibrium.

\begin{assumption}[Frame-bundle DIN]\label{ass:bundle}
There exists a smooth manifold $\Sigma\subset \mathrm{St}_1(F)$ with distinguished isotropic equilibrium $\chi\in\Sigma$ and a stationary symmetric collision model such that every $\sigma\in\Sigma$ generates a strongly continuous reversible one-parameter dynamics $U_{t,\sigma}$ on an inertial sector $\mcH$. The linearization \eqref{eq:omega} is well defined, linear, and injective on $T_\chi\Sigma$.
\end{assumption}

Assumption~\ref{ass:bundle} is the continuous-variable deformation of dynamics from informational nonequilibrium. Injectivity means that distinct infinitesimal reference-state imbalances induce distinct infinitesimal reversible dynamics.

\begin{assumption}[Smooth local completion of the single-state duality]\label{ass:smoothdual}
There exists an open neighborhood $U\subset\Sigma$ of $\chi$ such that, for every $\sigma\in U$, the hypotheses of Ref.~\cite{giannelli2026information} that produce a dual observable for $G_\sigma$ are satisfied. Let $\mcO_A$ denote the real vector space of self-adjoint operators on $\mcH$, and let $\mcOq:=\mcO_A/\R\I$ denote the quotient by additive scalars. Then there exists a map
\begin{equation}
\mathfrak A:U\longrightarrow \mcOq,
\end{equation}
with the following properties:
\begin{enumerate}
\item for each $\sigma\in U$, the class $\mathfrak A(\sigma)$ is the observable class selected by the Giannelli-Chiribella duality for $G_\sigma$;
\item the map $\mathfrak A$ is Fr\'echet differentiable at $\chi$;
\item for any representative $A(\sigma)\in\mathfrak A(\sigma)$, the reversible generator on the inertial sector is implemented as
\begin{equation}
G_\sigma(\rho)=-\frac{\ii}{\hbar}[A(\sigma),\rho].
\label{eq:quantumG}
\end{equation}
\end{enumerate}
\end{assumption}

Assumption~\ref{ass:smoothdual} is where the smooth continuous-variable completion of Ref.~\cite{giannelli2026information} enters. The quotient by scalars is natural because additive constants do not affect the reversible channel generated by Eq.~\eqref{eq:quantumG}.

\begin{proposition}[Derived directional observable-generator duality]\label{prop:duality}
Assumptions~\ref{ass:bundle} and \ref{ass:smoothdual} imply a well-defined linear injective map
\begin{equation}
\mcD:\omega(T_\chi\Sigma)\longrightarrow \mcOq
\end{equation}
such that, for every tangent vector $v\in T_\chi\Sigma$,
\begin{equation}
\omega(v)(\rho)=-\frac{\ii}{\hbar}[\mcD(\omega(v)),\rho].
\label{eq:derivedduality}
\end{equation}
\end{proposition}

The proof is given in Appendix~A. Proposition~\ref{prop:duality} identifies the local directional duality with the derivative of the single-state map.

\begin{assumption}[Local inertial geometry]\label{ass:geometry}
The tangent space at equilibrium decomposes as
\begin{equation}
T_\chi\Sigma=T_t\oplus T_{\mathrm{tr}}\oplus T_{\mathrm{rot}}\oplus T_{\mathrm{boost}},
\label{eq:tangentdecomp}
\end{equation}
with $\dim T_t=1$ and $\dim T_{\mathrm{tr}}=\dim T_{\mathrm{rot}}=\dim T_{\mathrm{boost}}=3$. Choosing bases $e_t$, $e_i$, $r_i$, and $b_i$, define observable classes by
\begin{align}
[H]&:=\mcD(\omega(e_t)), & [P_i]&:=\mcD(\omega(e_i)), \label{eq:defobsclass}\\
[J_i]&:=\mcD(\omega(r_i)), & [K_i]&:=\mcD(\omega(b_i)). \nonumber
\end{align}
Fix self-adjoint representatives $H,P_i,J_i,K_i$ of these classes once and for all. Let $\pi:\SU(2)\to\SOg(3)$ be the double-cover map, and denote the associated strongly continuous unitary representations by $T(t)$, $U(\vect{a})$, $R(u)$ with $u\in\SU(2)$, and $V(\vect{v})$. They satisfy the operational frame-composition laws
\begin{align}
T(t)U(\vect{a})T(-t) &= U(\vect{a}),\label{eq:TU}\\
R(u)U(\vect{a})R(u)^\dagger &= U(\pi(u)\vect{a}),\label{eq:RU}\\
R(u)V(\vect{v})R(u)^\dagger &= V(\pi(u)\vect{v}),\label{eq:RV}\\
T(t)V(\vect{v})T(-t) &= Z\!\left(-\tfrac12 t\,\abs{\vect{v}}^2\right)U(-t\vect{v})V(\vect{v}),\label{eq:TV}\\
T(t)R(u) &= R(u)T(t).\label{eq:TR}
\end{align}
\end{assumption}

Equation~\eqref{eq:TV} now records the standard finite time-boost composition in the present sign convention: conjugating a boost by time evolution produces a spatial translation by $-t\vect{v}$ together with the central Bargmann phase $Z(-\tfrac12 t\abs{\vect{v}}^2)$. Equation~\eqref{eq:TR} closes the time-rotation compatibility needed in the Hamiltonian derivation. Working with the double cover $\SU(2)$ rather than $\SOg(3)$ is what allows both integer and half-integer spin sectors. Importantly, Assumption~\ref{ass:geometry} specifies only the local geometry of frame changes; it does not assume a Hilbert-space realization of the Bargmann group.

\begin{assumption}[Covariant localization and fiducial focusing]\label{ass:loc}
There exists a POVM $E: \mathcal B(\R^3)\to \mathcal L(\mcH)$ such that for every Borel set $B\subseteq\R^3$,
\begin{align}
U(\vect{a})E(B)U(\vect{a})^\dagger &= E(B+\vect{a}),\label{eq:loc1}\\
R(u)E(B)R(u)^\dagger &= E(\pi(u) B),\label{eq:loc2}\\
V(\vect{v})E(B)V(\vect{v})^\dagger &= E(B),\label{eq:loc3}
\end{align}
and, in addition, for every $r>0$ and every $\varepsilon>0$ there exists a normalized state $\psi_{r,\varepsilon}\in\mcH$ such that
\begin{equation}
\langle \psi_{r,\varepsilon},E(B_r(0))\psi_{r,\varepsilon}\rangle\ge 1-\varepsilon.
\label{eq:focusing}
\end{equation}
Here $B_r(0)$ denotes the open ball of radius $r$ around the origin. Finally, we assume the standard covariant-localization smearing structure: there exist a sharp transitive $\widetilde{E}(3)$-system of imprimitivity $\Pi:\mathcal B(\R^3)\to\mathcal L(\mcH)$ for the spin-cover Euclidean group $\widetilde{E}(3):=\R^3\rtimes\SU(2)$ and a rotation-invariant probability measure $\rho$ on $\R^3$ such that
\begin{equation}
E(B)=\int_{\R^3}\rho(B-\vect{x})\,\Pi(\dd \vect{x}).
\label{eq:fuzzyX}
\end{equation}
Moreover, the covariant localization system $(U,R,E)$ is irreducible, i.e. it has no nontrivial closed subspace invariant under all $U(\vect{a})$, $R(u)$, and $E(B)$.
\end{assumption}

The first three relations state that position statistics transform as they should under translations, covered rotations, and boosts at a fixed instant. The focusing clause is weaker than imposing the norm-one property for every open set; it demands only that arbitrarily fine localization be achievable in one fiducial frame. The smearing clause is the standard structural input for covariant localization observables; in the Schr\"odinger realization it reduces to the Carmeli-Heinonen-Toigo classification \cite{carmeli2004position}.

\begin{proposition}[Fiducial focusing implies global norm-one localization]\label{prop:focusing}
Assumption~\ref{ass:loc} implies that for every nonempty open set $\mcO\subseteq\R^3$ one has
\begin{equation}
\norm{E(\mcO)}=1.
\label{eq:normone}
\end{equation}
\end{proposition}

The proof is given in Appendix~B. Starting from focusing in one fiducial frame is enough to recover the norm-one property for every nonempty open set.

\begin{assumption}[Boost-translation holonomy]\label{ass:holonomy}
There exists a strongly continuous one-parameter unitary subgroup $Z(\lambda)$ on $\mcH$, commuting with $T(t)$, $U(\vect{a})$, $V(\vect{v})$, and $R(u)$, such that for infinitesimal loops one has
\begin{equation}
V(\delta\vect{v})U(\delta\vect{a})V(-\delta\vect{v})U(-\delta\vect{a})
=
Z\!\left(\delta\vect{v}\cdot\delta\vect{a}+o(\abs{\delta\vect{v}}\abs{\delta\vect{a}})\right).
\label{eq:holonomy}
\end{equation}
The subgroup $Z$ acts nontrivially on the sector.
\end{assumption}

Assumption~\ref{ass:holonomy} is the replacement for a Bargmann-first postulate. It says that the only irreducible curvature of inertial frame space is the holonomy picked up by a boost-translation loop. The observable dual to the generator of $Z$ will be identified with mass.

\begin{assumption}[Irreducible inertial sector]\label{ass:irr}
The reversible transformations generated by $T(t)$, $U(\vect{a})$, $V(\vect{v})$, $R(u)$, and $Z(\lambda)$ act irreducibly on a separable Hilbert space $\mcH$.
\end{assumption}

\begin{theorem}[Main theorem]\label{thm:main}
Assumptions~\ref{ass:bundle}--\ref{ass:irr} imply that there exist unique $m>0$ and $s\in\{0,\tfrac12,1,\ldots\}$ and a unitary equivalence
\begin{equation}
\mcH\cong L^2(\R^3)\otimes\C^{2s+1}
\end{equation}
for which localization is canonical, the translational observable is the canonical momentum, the boost holonomy is a scalar mass, and total angular momentum splits into orbital plus spin. More explicitly, with $\pi:\SU(2)\to\SOg(3)$ the double-cover map and $u\in\SU(2)$,
\begin{widetext}
\begin{align}
(E(B)\psi)(\vect{x}) &= \chi_B(\vect{x})\,\psi(\vect{x}),\label{eq:explicitE}\\
(U(\vect{a})\psi)(\vect{x}) &= \psi(\vect{x}-\vect{a}),\qquad (V(\vect{v})\psi)(\vect{x}) = \exp\!\left(-\frac{\ii}{\hbar}m\,\vect{v}\!\cdot\!\vect{x}\right)\psi(\vect{x}),\label{eq:explicitUV}\\
(R(u)\psi)(\vect{x}) &= D^{(s)}(u)\,\psi(\pi(u)^{-1}\vect{x}),\label{eq:explicitR}\\
(\widehat{T(t)\psi})(\vect{p}) &= \exp\!\left[-\frac{\ii t}{\hbar}\left(\frac{\vect{p}^{\,2}}{2m}+E_0\right)\right]\widehat\psi(\vect{p}).\label{eq:explicitT}
\end{align}
\end{widetext}
with observables
\begin{align}
P_i &= -\ii\hbar\,\partial_i, & M&=m\I, & K_i&=mX_i,\label{eq:gens1}\\
J_i &= \epsilon_{ijk}X_jP_k+S_i, & [X_i,P_j] &= \ii\hbar\,\delta_{ij}\I, & P_i&=m\dot X_i.\label{eq:gens2}
\end{align}
Moreover, the joint spectral measure of $\vect{P}$ is absolutely continuous with respect to Lebesgue measure and has support $\R^3$, and the unique sharp momentum observable satisfying translation invariance, boost covariance, and rotational covariance is the spectral measure of $\vect{P}$.
\end{theorem}

In this representation the central subgroup acts as $Z(\lambda)=\exp[-(\ii/\hbar)m\lambda]\I$. The constant $E_0$ is the residual additive freedom in the temporal observable class. It has no effect on the reversible channel generated by $H$.

\paragraph*{Scope of the assumptions.}
We do not start from a projective Bargmann representation. The inputs are a smooth family of reference states, a smooth local completion of the single-state duality, a localization POVM with fiducial focusing and the standard smearing structure, inertial composition laws with the $\SU(2)$ spin cover of rotations, and a holonomy condition. From these one recovers the Schr\"odinger representation, the canonical commutation relations, scalar mass, and the orbital-spin decomposition. The argument remains conditional because the smooth completion, the localization observable, and the holonomy input are not derived here.

\section{Outline of the argument}

The appendices contain the detailed derivations. The main line of argument has five steps.

\subsection{Local directional duality from smooth completion}

Appendix~A differentiates the map $\sigma\mapsto \mathfrak A(\sigma)$ at $\chi$. Once the single-state assignment extends to a neighborhood of equilibrium, its derivative gives a linear map from tangent directions of the reference-state manifold to observable classes. Equation~\eqref{eq:derivedduality} then follows by differentiating the commutator form of the generators.

\subsection{Sharp localization from fiducial focusing and covariance}

Appendix~B proves Proposition~\ref{prop:focusing} by translating the fiducial concentrating states. If a state can be made arbitrarily well localized in a small ball around the origin, translation covariance moves the same construction to any other point, and the norm-one property \eqref{eq:normone} follows for every nonempty open region.

Once Eq.~\eqref{eq:normone} is available, the structural part of Assumption~\ref{ass:loc} writes the localization POVM as a rotation-invariant smearing of a sharp transitive localization PVM. Proposition~\ref{prop:sharpE} in Appendix~B shows that the norm-one property rules out every nontrivial smearing and forces $\rho=\delta_0$. Hence $E$ is sharp. Let $\vect{X}$ denote the self-adjoint triple with spectral measure $E$. In particular,
\begin{equation}
U(\vect{a})^\dagger X_i U(\vect{a})=X_i+a_i\I.
\label{eq:Xcov}
\end{equation}
Differentiating Eq.~\eqref{eq:Xcov} at $\vect{a}=0$ gives the canonical commutator
\begin{equation}
[X_i,P_j]=\ii\hbar\delta_{ij}\I.
\label{eq:CCR}
\end{equation}
Thus the translation direction singles out the canonical momentum conjugate to the sharp position observable.

\subsection{Euclidean imprimitivity fixes the translation and rotation sectors}

Once localization is sharp, $(U,R,E)$ is a transitive irreducible system of imprimitivity for the spin-cover Euclidean group $\widetilde{E}(3)=\R^3\rtimes \SU(2)$. Mackey's theorem therefore yields a unique induced representation on $L^2(\R^3)\otimes \mcK$, where $\mcK$ carries an irreducible representation of the little group at the origin, namely $\SU(2)$ \cite{mackey1968induced}. Hence $\mcK\cong \C^{2s+1}$ for a unique spin $s$, and one may choose the representation so that Eqs.~\eqref{eq:explicitE} and \eqref{eq:explicitR} hold. Stone's theorem then identifies the translational observable with
\begin{equation}
P_i=-\ii\hbar\partial_i,
\end{equation}
whose spectrum is all of $\R$ and purely absolutely continuous. Differentiating the rotation representation gives
\begin{equation}
\vect{J}=\vect{X}\times\vect{P}+\vect{S}.
\label{eq:JdecompMain}
\end{equation}
So the translation and rotation generators already take their standard canonical form.

\subsection{Mass from holonomy and boosts from sharp localization}

Let $M$ denote the observable dual to the generator of the central subgroup $Z(\lambda)$. Expanding Eq.~\eqref{eq:holonomy} to first nontrivial order gives
\begin{equation}
[K_i,P_j]=\ii\hbar\delta_{ij}M.
\label{eq:KPM}
\end{equation}
Because $Z(\lambda)$ commutes with the whole inertial-frame action, Schur's lemma on the irreducible sector implies
\begin{equation}
M=m\I
\label{eq:massscalar}
\end{equation}
for some nonzero real number $m$. Reversing the orientation of the boost parameter if necessary, we take $m>0$.

Next, boost invariance of sharp localization means that $V(\vect{v})$ lies in the commutant of the canonical position PVM. Therefore in the position representation it must act by multiplication with a measurable unitary matrix-valued function,
\begin{equation}
(V(\vect{v})\psi)(\vect{x})=W_{\vect{v}}(\vect{x})\psi(\vect{x}).
\end{equation}
Using the exact Weyl relation generated by Eq.~\eqref{eq:KPM}, one finds
\begin{equation}
W_{\vect{v}}(\vect{x}+\vect{a})=
\exp\!\left(-\frac{\ii}{\hbar}m\,\vect{v}\!\cdot\!\vect{a}\right)W_{\vect{v}}(\vect{x}),
\end{equation}
so
\begin{equation}
W_{\vect{v}}(\vect{x})=\exp\!\left(-\frac{\ii}{\hbar}m\,\vect{v}\!\cdot\!\vect{x}\right)C(\vect{v})
\end{equation}
with $C(\vect{v})$ acting only on spin. Rotational covariance forces the finite-dimensional character family $C(\vect{v})$ to be trivial, hence $C(\vect{v})=\I$. Therefore boosts are precisely the Galilei phases of Eq.~\eqref{eq:explicitUV}, and differentiation gives
\begin{equation}
K_i=mX_i.
\label{eq:KeqmXmain}
\end{equation}
This identifies boosts with the usual Galilei phase factors at $t=0$ and fixes $K_i=mX_i$.

\subsection{Free Hamiltonian and mechanical momentum}

Appendix~E uses the time-translation, time-boost, and time-rotation compatibility relations \eqref{eq:TU}, \eqref{eq:TV}, and \eqref{eq:TR} to derive
\begin{equation}
[H,P_i]=0,
\qquad
[K_i,H]=\ii\hbar P_i,
\qquad
[H,J_i]=0.
\label{eq:timecommmain}
\end{equation}
Combining Eq.~\eqref{eq:timecommmain} with Eq.~\eqref{eq:KeqmXmain} yields
\begin{equation}
[X_i,H]=\frac{\ii\hbar}{m}P_i,
\end{equation}
and therefore, in the Heisenberg picture,
\begin{equation}
P_i=m\dot X_i.
\label{eq:mvelocitymain}
\end{equation}
To determine the Hamiltonian, define $C:=2mH-\vect{P}^{\,2}$. Using Eqs.~\eqref{eq:KPM} and \eqref{eq:timecommmain} with $M=m\I$, one checks that $C$ commutes with the full irreducible inertial action. By Schur's lemma, $C=2mE_0\I$, so
\begin{equation}
H=\frac{\vect{P}^{\,2}}{2m}+E_0\I.
\label{eq:freeHmain}
\end{equation}
Together with Eq.~\eqref{eq:JdecompMain}, this proves Theorem~\ref{thm:main}.

\section{Consequences and discussion}

Two features of the construction matter most. The directional observable-generator map comes from differentiating the single-state duality, and sharp localization comes from fiducial focusing plus covariance together with the standard smearing structure of covariant localization observables. The extra input is therefore concentrated in the smooth completion, the localization observable, the inertial composition laws, and the holonomy condition.

Momentum is fixed in two compatible ways. Proposition~\ref{prop:duality} ties it to translational imbalance, while Propositions~\ref{prop:sharpE} and \ref{prop:momentumobs} show that the same operator is the unique sharp momentum observable compatible with translation, boost, and rotation covariance.

Mass enters through the holonomy of boost-translation loops. In geometric terms, the connection $\omega$ in Eq.~\eqref{eq:omega} has a curvature component in the boost-translation plane, and on an irreducible sector that component reduces to $m\I$. This is the place where Bargmann's central charge appears in the present framework.

The same local duality also yields operational speed bounds for spatial and rotational motions. For a unit vector $\hat{\vect{n}}$, the translated and rotated states
\begin{align}
\rho_a &= U(a\hat{\vect{n}})\rho U(a\hat{\vect{n}})^\dagger,\\
\rho_\theta &= \exp\!\left(-\frac{\ii}{\hbar}\theta\,\hat{\vect{n}}\!\cdot\!\vect{J}\right)\rho\exp\!\left(+\frac{\ii}{\hbar}\theta\,\hat{\vect{n}}\!\cdot\!\vect{J}\right)
\end{align}
obey the same Euclidean distinguishability bound proved in Ref.~\cite{giannelli2026information}, now with $P_{\hat{\vect{n}}}$ and $J_{\hat{\vect{n}}}$ in place of the Hamiltonian. Momentum and angular momentum therefore control the rate at which the state moves along the corresponding symmetry orbits, in the same sense that energy controls time evolution.

What remains open is to obtain the smooth local completion of the single-state duality, the localization POVM together with its smearing structure, and the holonomy condition from a genuinely infinite-dimensional extension of the operational axioms of Ref.~\cite{giannelli2026information}. Within the present assumptions, however, the Galilean one-particle sector already follows without starting from a projective Bargmann representation.

\setcounter{section}{0}
\renewcommand{\thesection}{\Alph{section}}
\makeatletter
\@addtoreset{equation}{section}
\makeatother
\renewcommand{\theequation}{\thesection\arabic{equation}}
\setcounter{equation}{0}

\section{Directional duality from the local state family}

This appendix proves Proposition~\ref{prop:duality} under the assumption that the single-state observable assignment extends smoothly to a neighborhood of the isotropic equilibrium.

\begin{proof}[Proof of Proposition~\ref{prop:duality}]
Let $v\in T_\chi\Sigma$, and choose a smooth curve $\varepsilon\mapsto \sigma_\varepsilon$ in $\Sigma$ such that $\sigma_0=\chi$ and $\dot\sigma_0=v$. Since $\mathfrak A$ is Fr\'echet differentiable at $\chi$, the derivative
\begin{equation}
\dd\mathfrak A_\chi(v):=\left.\frac{\dd}{\dd\varepsilon}\right\rvert_{0}\mathfrak A(\sigma_\varepsilon)\in\mcOq
\end{equation}
depends only on the tangent vector $v$, not on the chosen curve. Because $\omega$ is injective by Assumption~\ref{ass:bundle}, every element of $\omega(T_\chi\Sigma)$ has a unique preimage $v$. We may therefore define
\begin{equation}
\mcD(\omega(v)):=\dd\mathfrak A_\chi(v).
\label{eq:defDappendix}
\end{equation}
Linearity of $\mcD$ follows from linearity of both $\omega$ and the differential of $\mathfrak A$.

To obtain Eq.~\eqref{eq:derivedduality}, choose representatives $A(\sigma_\varepsilon)\in\mathfrak A(\sigma_\varepsilon)$. By Eq.~\eqref{eq:quantumG},
\begin{equation}
G_{\sigma_\varepsilon}(\rho)=-\frac{\ii}{\hbar}[A(\sigma_\varepsilon),\rho].
\end{equation}
Differentiating at $\varepsilon=0$ gives
\begin{equation}
\omega(v)(\rho)
=
-\frac{\ii}{\hbar}\left[\left.\frac{\dd}{\dd\varepsilon}\right\rvert_0 A(\sigma_\varepsilon),\rho\right].
\end{equation}
If a different representative is chosen, $A(\sigma_\varepsilon)$ changes by a scalar function $c(\varepsilon)\I$, whose derivative has vanishing commutator with every state. Hence the commutator depends only on the class $\dd\mathfrak A_\chi(v)=\mcD(\omega(v))$, and Eq.~\eqref{eq:derivedduality} follows.

Finally, suppose $\mcD(\omega(v))=0$ in the quotient $\mcOq$. Then $\dd\mathfrak A_\chi(v)$ is represented by a scalar operator, so Eq.~\eqref{eq:derivedduality} gives $\omega(v)=0$. Injectivity of $\omega$ then implies $v=0$. Therefore $\mcD$ is injective.
\end{proof}

Because the codomain is a quotient by scalars, the self-adjoint operators $H$, $\vect{P}$, $\vect{J}$, and $\vect{K}$ used in the main text should be understood as chosen representatives of the corresponding observable classes. All commutator relations derived later are independent of those scalar choices.

\section{Fiducial focusing and sharp covariant localization}

This section proves Proposition~\ref{prop:focusing} and then shows that the resulting localization observable is sharp.

\begin{proof}[Proof of Proposition~\ref{prop:focusing}]
Let $\mcO\subseteq\R^3$ be a nonempty open set. Choose $\vect{x}_0\in\mcO$ and $r>0$ such that the open ball $B_r(\vect{x}_0)$ is contained in $\mcO$. By Assumption~\ref{ass:loc}, for every $\varepsilon>0$ there exists a normalized state $\psi_{r,\varepsilon}$ such that
\begin{equation}
\langle \psi_{r,\varepsilon},E(B_r(0))\psi_{r,\varepsilon}\rangle\ge 1-\varepsilon.
\end{equation}
Define $\phi_{r,\varepsilon}:=U(\vect{x}_0)\psi_{r,\varepsilon}$. Translation covariance gives
\begin{equation}
\langle \phi_{r,\varepsilon},E(B_r(\vect{x}_0))\phi_{r,\varepsilon}\rangle
=
\langle \psi_{r,\varepsilon},E(B_r(0))\psi_{r,\varepsilon}\rangle
\ge 1-\varepsilon.
\end{equation}
Since $B_r(\vect{x}_0)\subseteq\mcO$, positivity of the POVM yields $E(\mcO)\ge E(B_r(\vect{x}_0))$, and therefore
\begin{equation}
\langle \phi_{r,\varepsilon},E(\mcO)\phi_{r,\varepsilon}\rangle\ge 1-\varepsilon.
\end{equation}
Taking the supremum over unit vectors and then letting $\varepsilon\to0$ gives $\norm{E(\mcO)}=1$.
\end{proof}

\begin{proposition}[Norm-one localization forces sharpness]\label{prop:sharpE}
Let $E$ satisfy Assumption~\ref{ass:loc}. If, in addition, $\norm{E(\mcO)}=1$ for every nonempty open set $\mcO\subseteq\R^3$, then the smearing measure in Eq.~\eqref{eq:fuzzyX} is $\rho=\delta_0$. Consequently $E=\Pi$ is sharp.
\end{proposition}

\begin{proof}
Suppose first that $\rho$ is not a point measure. Then its support contains two distinct points $\vect{x}_1\neq \vect{x}_2$. Choose disjoint open balls $B_{\eta}(\vect{x}_1)$ and $B_{\eta}(\vect{x}_2)$ around them, with $\eta>0$ so small that $4\eta<\abs{\vect{x}_1-\vect{x}_2}$. Because $\vect{x}_1$ and $\vect{x}_2$ belong to the support, both balls have strictly positive $\rho$-measure. Set
\begin{equation}
\delta:=\min\bigl\{\rho(B_{\eta}(\vect{x}_1)),\rho(B_{\eta}(\vect{x}_2))\bigr\}>0,
\end{equation}
and let $\mcO:=B_{\eta}(0)$. Every translate $\mcO+\vect{y}$ has diameter $2\eta$, so it cannot intersect both balls $B_{\eta}(\vect{x}_1)$ and $B_{\eta}(\vect{x}_2)$: if it did, it would contain points whose distance is at least $\abs{\vect{x}_1-\vect{x}_2}-2\eta>2\eta$, impossible for a set of diameter $2\eta$. Hence for every $\vect{y}\in\R^3$ at least one of the two balls is disjoint from $\mcO+\vect{y}$, and therefore
\begin{equation}
\rho(\mcO-\vect{y})=\rho(\mcO+(-\vect{y}))\le 1-\delta.
\end{equation}
Consequently,
\begin{equation}
\sup_{\vect{y}\in\R^3}\rho(\mcO-\vect{y})\le 1-\delta<1.
\end{equation}
Using Eq.~\eqref{eq:fuzzyX}, any normalized vector $\psi$ satisfies
\begin{equation}
\begin{aligned}
\langle\psi,E(\mcO)\psi\rangle
&=
\int_{\R^3}\rho(\mcO-\vect{x})\,\langle\psi,\Pi(\dd\vect{x})\psi\rangle \\
&\le \sup_{\vect{y}}\rho(\mcO-\vect{y})\le 1-\delta<1,
\end{aligned}
\end{equation}
which contradicts $\norm{E(\mcO)}=1$. Therefore $\rho$ must be a point measure, say $\rho=\delta_{\vect{x}_0}$. Rotation invariance then forces $\vect{x}_0=0$. Hence $\rho=\delta_0$ and $E=\Pi$, which is projection valued.
\end{proof}

An immediate consequence is that $E$ is a projection-valued measure. Therefore there exists a self-adjoint position triple $\vect{X}$ with spectral measure $E$, and Eq.~\eqref{eq:Xcov} holds. Differentiating Eq.~\eqref{eq:Xcov} at the identity of the translation subgroup yields Eq.~\eqref{eq:CCR}.

\section{Mackey reconstruction of the translation and rotation sectors}

\begin{proposition}[Euclidean imprimitivity and spin]\label{prop:mackey}
Let $(U,R,E)$ satisfy Assumptions~\ref{ass:geometry} and \ref{ass:loc}. Then there exists a unique $s\in\{0,\tfrac12,1,\ldots\}$ and a unitary equivalence
\begin{equation}
\mcH\cong L^2(\R^3)\otimes\C^{2s+1}
\end{equation}
under which Eqs.~\eqref{eq:explicitE} and \eqref{eq:explicitR} hold.
\end{proposition}

\begin{proof}
By Proposition~\ref{prop:focusing}, the localization observable obeys the global norm-one property, and by Proposition~\ref{prop:sharpE} it is therefore sharp. Hence $(U,R,E)$ is a transitive irreducible system of imprimitivity for the spin-cover Euclidean group $\widetilde{E}(3)=\R^3\rtimes\SU(2)$. Mackey's theorem therefore identifies the representation as one induced from an irreducible unitary representation of the stabilizer of the origin, namely $\SU(2)$ \cite{mackey1968induced}. The irreducible finite-dimensional unitary representations of $\SU(2)$ are labeled by spin $s\in\{0,\tfrac12,1,\ldots\}$, with carrier space $\C^{2s+1}$. In the induced realization one has
\begin{align}
(E(B)\psi)(\vect{x}) &= \chi_B(\vect{x})\psi(\vect{x}),\\
(U(\vect{a})\psi)(\vect{x}) &= \psi(\vect{x}-\vect{a}),\\
(R(u)\psi)(\vect{x}) &= D^{(s)}(u)\psi(\pi(u)^{-1}\vect{x}),
\end{align}
which is the desired form.
\end{proof}

\begin{corollary}[Canonical momentum and angular momentum]\label{cor:pj}
In the representation of Proposition~\ref{prop:mackey}, the generator of translations is
\begin{equation}
P_i=-\ii\hbar\partial_i,
\end{equation}
with spectrum $\spec(P_i)=\R$, and the total angular momentum is
\begin{equation}
J_i=\epsilon_{ijk}X_jP_k+S_i.
\end{equation}
The spin operators satisfy the usual $\mathfrak{su}(2)$ commutation relations and commute with all $X_i$ and $P_i$.
\end{corollary}

\begin{proof}
The form of $U(\vect{a})$ in Proposition~\ref{prop:mackey} is the standard translation representation, so Stone's theorem gives $P_i=-\ii\hbar\partial_i$. Under the Fourier transform on $L^2(\R^3)\otimes\C^{2s+1}$, the triple $\vect{P}$ becomes multiplication by $\vect{p}$. Hence its joint spectral measure is multiplication by characteristic functions in momentum space, is absolutely continuous with respect to Lebesgue measure, and has support all of $\R^3$. In particular, each component $P_i$ has spectrum $\R$ and is absolutely continuous. Differentiating the rotation representation at the identity yields
\begin{equation}
\vect{J}=\vect{L}+\vect{S},\qquad \vect{L}:=\vect{X}\times\vect{P},
\end{equation}
with $\vect{S}$ acting on the internal factor only.
\end{proof}

\section{Mass as central holonomy and the form of boosts}

Here we show that the holonomy subgroup is scalar and recover the standard form of the boost operators.

\begin{lemma}[From infinitesimal holonomy to commutators]\label{lem:holonomycomm}
Assumption~\ref{ass:holonomy} implies
\begin{equation}
[K_i,P_j]=\ii\hbar\delta_{ij}M,
\label{eq:KPappendix}
\end{equation}
where $M$ is the self-adjoint generator of the central subgroup $Z(\lambda)$.
\end{lemma}

\begin{proof}
Let
\begin{equation}
A:= -\frac{\ii}{\hbar}\delta v_i K_i,
\qquad
B:= -\frac{\ii}{\hbar}\delta a_j P_j.
\end{equation}
Define the loop operator
\begin{equation}
\mathcal W(\delta\vect{v},\delta\vect{a})
:=
V(\delta\vect{v})U(\delta\vect{a})V(-\delta\vect{v})U(-\delta\vect{a}).
\end{equation}
Then
\begin{equation}
\mathcal W(\delta\vect{v},\delta\vect{a})
=
\e^A\e^B\e^{-A}\e^{-B}
=
\exp\!\bigl([A,B]+o(\abs{\delta\vect{v}}\abs{\delta\vect{a}})\bigr).
\end{equation}
Since
\begin{equation}
[A,B]= -\frac{1}{\hbar^2}\delta v_i\delta a_j [K_i,P_j],
\end{equation}
and
\begin{equation}
Z(\lambda)=\I-\frac{\ii}{\hbar}\lambda M+o(\lambda),
\end{equation}
Eq.~\eqref{eq:holonomy} gives
\begin{equation}
-\frac{1}{\hbar^2}\delta v_i\delta a_j [K_i,P_j]
=
-\frac{\ii}{\hbar}(\delta\vect{v}\cdot\delta\vect{a})M+o(\abs{\delta\vect{v}}\abs{\delta\vect{a}}).
\end{equation}
Since this holds for arbitrary infinitesimal vectors, Eq.~\eqref{eq:KPappendix} follows.
\end{proof}

\begin{proposition}[Mass is scalar on an irreducible sector]\label{prop:massscalar}
The operator $M$ of Lemma~\ref{lem:holonomycomm} is of the form $M=m\I$ with $m\neq 0$. By reversing the sign convention for the boost parameter if necessary, one may choose $m>0$.
\end{proposition}

\begin{proof}
By Assumption~\ref{ass:holonomy}, the one-parameter subgroup $Z(\lambda)$ commutes with the full inertial action. Hence its generator $M$ commutes with the full irreducible inertial action generated by $T$, $U$, $V$, $R$, and $Z$. Schur's lemma then implies $M=m\I$ for some real number $m$. The subgroup acts nontrivially by assumption, so $m\neq 0$. Replacing $\vect{v}$ with $-\vect{v}$ flips the sign of $m$, so we choose $m>0$.
\end{proof}

\begin{lemma}[Exact Weyl relation]\label{lem:weyl}
With $M=m\I$, the boost and translation operators obey the exact Weyl relation
\begin{equation}
V(\vect{v})U(\vect{a})=\exp\!\left(-\frac{\ii}{\hbar}m\,\vect{v}\!\cdot\!\vect{a}\right)U(\vect{a})V(\vect{v}).
\label{eq:weyl}
\end{equation}
\end{lemma}

\begin{proof}
Set $A:=-(\ii/\hbar)\vect{v}\cdot\vect{K}$ and $B:=-(\ii/\hbar)\vect{a}\cdot\vect{P}$. By Proposition~\ref{prop:massscalar} and Lemma~\ref{lem:holonomycomm},
\begin{equation}
[A,B]=-\frac{\ii}{\hbar}m\,\vect{v}\!\cdot\!\vect{a}\;\I,
\end{equation}
which is central. Therefore the Baker-Campbell-Hausdorff series truncates after the first commutator:
\begin{equation}
\e^A\e^B=\e^{[A,B]}\e^B\e^A.
\end{equation}
Substituting back the definitions of $A$ and $B$ gives Eq.~\eqref{eq:weyl}.
\end{proof}

\begin{proposition}[Boosts are generated by $m\vect{X}$]\label{prop:boostform}
In the representation of Proposition~\ref{prop:mackey},
\begin{equation}
(V(\vect{v})\psi)(\vect{x})=\exp\!\left(-\frac{\ii}{\hbar}m\,\vect{v}\!\cdot\!\vect{x}\right)\psi(\vect{x}),
\end{equation}
and therefore
\begin{equation}
K_i=mX_i.
\end{equation}
\end{proposition}

\begin{proof}
Because boosts leave the sharp localization observable invariant, Eq.~\eqref{eq:loc3} shows that $V(\vect{v})$ belongs to the commutant of the canonical position PVM. Thus in the position representation,
\begin{equation}
(V(\vect{v})\psi)(\vect{x})=W_{\vect{v}}(\vect{x})\psi(\vect{x})
\end{equation}
for some measurable unitary matrix-valued function $W_{\vect{v}}(\vect{x})$ acting on the spin space.

Apply the exact Weyl relation~\eqref{eq:weyl} to a wavefunction. Since translations act by $(U(\vect{a})\psi)(\vect{x})=\psi(\vect{x}-\vect{a})$, one obtains
\begin{equation}
W_{\vect{v}}(\vect{x})\psi(\vect{x}-\vect{a})
=
\exp\!\left(-\frac{\ii}{\hbar}m\,\vect{v}\!\cdot\!\vect{a}\right)W_{\vect{v}}(\vect{x}-\vect{a})\psi(\vect{x}-\vect{a}),
\end{equation}
whence
\begin{equation}
W_{\vect{v}}(\vect{x}+\vect{a})=\exp\!\left(-\frac{\ii}{\hbar}m\,\vect{v}\!\cdot\!\vect{a}\right)W_{\vect{v}}(\vect{x}).
\label{eq:Wfunctional}
\end{equation}
Setting $\vect{x}=0$ in Eq.~\eqref{eq:Wfunctional} yields
\begin{equation}
W_{\vect{v}}(\vect{x})=\exp\!\left(-\frac{\ii}{\hbar}m\,\vect{v}\!\cdot\!\vect{x}\right)C(\vect{v}),
\label{eq:WC}
\end{equation}
where $C(\vect{v})$ acts only on spin.

The family $C(\vect{v})$ is a strongly continuous unitary representation of the additive group $\R^3$ on the finite-dimensional spin space. Hence all $C(\vect{v})$ commute and can be simultaneously diagonalized. Let $\{\eta_\alpha\}$ be a common eigenbasis, with
\begin{equation}
C(\vect{v})\eta_\alpha=\exp\!\left(\frac{\ii}{\hbar}\vect{\beta}_\alpha\!\cdot\!\vect{v}\right)\eta_\alpha.
\end{equation}
Rotational covariance,
\begin{equation}
R(u)V(\vect{v})R(u)^\dagger=V(\pi(u)\vect{v}),
\end{equation}
forces the finite set $\{\vect{\beta}_\alpha\}$ to be $\SOg(3)$-invariant. The only finite rotation-invariant subset of $\R^3$ is $\{\vect{0}\}$. Therefore all $\vect{\beta}_\alpha$ vanish and $C(\vect{v})=\I$. Equation~\eqref{eq:WC} then gives
\begin{equation}
(V(\vect{v})\psi)(\vect{x})=\exp\!\left(-\frac{\ii}{\hbar}m\,\vect{v}\!\cdot\!\vect{x}\right)\psi(\vect{x}).
\end{equation}
Differentiating at $\vect{v}=0$ yields $K_i=mX_i$.
\end{proof}

\section{Hamiltonian and mechanical momentum}

\begin{lemma}[Time compatibility relations]\label{lem:timecompat}
Equation~\eqref{eq:TV} implies
\begin{equation}
[K_i,H]=\ii\hbar P_i,
\end{equation}
equivalently $[H,K_i]=-\ii\hbar P_i$. Equation~\eqref{eq:TU} implies $[H,P_i]=0$. Equation~\eqref{eq:TR} implies $[H,J_i]=0$.
\end{lemma}

\begin{proof}
Differentiate Eq.~\eqref{eq:TV} with respect to $t$ at $t=0$. The left-hand side gives $-(\ii/\hbar)[H,V(\vect{v})]$. The right-hand side is the derivative of
\begin{equation}
Z\!\left(-\tfrac12 t\,\abs{\vect{v}}^2\right)U(-t\vect{v})V(\vect{v})
\end{equation}
at $t=0$, namely
\begin{equation}
\frac{\ii}{\hbar}\left(\frac12 \abs{\vect{v}}^2 M+\vect{v}\cdot\vect{P}\right)V(\vect{v}).
\end{equation}
The $M$-term is quadratic in $\vect{v}$, so it disappears when one compares linear terms in $\vect{v}$ at the identity. Hence $[H,K_i]=-\ii\hbar P_i$, equivalently $[K_i,H]=\ii\hbar P_i$. Likewise, differentiating Eq.~\eqref{eq:TU} with respect to $t$ at $0$ shows that $[H,P_i]=0$.

Finally, Eq.~\eqref{eq:TR} implies $T(t)R(u)T(-t)=R(u)$ for every $t$ and $u\in\SU(2)$. Differentiating with respect to $t$ at $0$ yields $[H,R(u)]=0$ for all $u$, and differentiating at the identity of $\SU(2)$ gives $[H,J_i]=0$.
\end{proof}

\begin{proposition}[Free-particle Hamiltonian]\label{prop:freeH}
On the irreducible sector,
\begin{equation}
H=\frac{\vect{P}^{\,2}}{2m}+E_0\I
\end{equation}
for some real constant $E_0$.
\end{proposition}

\begin{proof}
By Proposition~\ref{prop:boostform}, $K_i=mX_i$. Combining this with Lemma~\ref{lem:timecompat} gives
\begin{equation}
[X_i,H]=\frac{\ii\hbar}{m}P_i.
\label{eq:XHappendix}
\end{equation}
Now define
\begin{equation}
C:=2mH-\vect{P}^{\,2}.
\end{equation}
From $[H,P_i]=0$ we immediately get $[C,P_i]=0$. Next,
\begin{align}
[K_i,\vect{P}^{\,2}] &= [K_i,P_j]P_j+P_j[K_i,P_j]
=2\ii\hbar m P_i,
\end{align}
using Proposition~\ref{prop:massscalar}. Therefore
\begin{equation}
[K_i,C]=2m[K_i,H]-[K_i,\vect{P}^{\,2}]=0.
\end{equation}
Lemma~\ref{lem:timecompat} also gives $[H,J_i]=0$, while of course $[J_i,\vect{P}^{\,2}]=0$, so $[J_i,C]=0$. Finally, $C$ commutes with $H$ by construction, hence with $T(t)$, and it also commutes with the central subgroup $Z(\lambda)$. Thus $C$ commutes with the full irreducible inertial action generated by $T$, $U$, $V$, $R$, and $Z$. Schur's lemma yields $C=2mE_0\I$, proving the claim.
\end{proof}

\begin{corollary}[Mechanical momentum]\label{cor:mechP}
The translational observable is the mechanical momentum:
\begin{equation}
P_i=m\dot X_i.
\end{equation}
\end{corollary}

\begin{proof}
In the Heisenberg picture,
\begin{equation}
\dot X_i=\frac{\ii}{\hbar}[H,X_i].
\end{equation}
Using Eq.~\eqref{eq:XHappendix} gives $\dot X_i=P_i/m$.
\end{proof}

\section{Momentum observables and covariance}

For completeness we record the covariance statement for momentum observables used in the main theorem.

\begin{proposition}[Sharp momentum observable]\label{prop:momentumobs}
Let $F$ be a POVM on $\R^3$ satisfying
\begin{align}
U(\vect{a})F(\Delta)U(\vect{a})^\dagger &= F(\Delta),\label{eq:momcov1}\\
V(\vect{v})F(\Delta)V(\vect{v})^\dagger &= F(\Delta-m\vect{v}),\label{eq:momcov2}\\
R(u)F(\Delta)R(u)^\dagger &= F(\pi(u)\Delta),\label{eq:momcov3}
\end{align}
for every Borel set $\Delta\subseteq\R^3$. Then there exists a rotation-invariant probability measure $\nu$ such that
\begin{equation}
F(\Delta)=\int_{\R^3}\nu(\Delta-\vect{p})\,\Pi_{\vect{P}}(\dd\vect{p}),
\end{equation}
where $\Pi_{\vect{P}}$ is the spectral measure of the canonical momentum operator. If $F$ is sharp, then $\nu=\delta_0$ and $F=\Pi_{\vect{P}}$.
\end{proposition}

\begin{proof}
This is the momentum-space counterpart of Proposition~\ref{prop:sharpE} and again follows from the classification of covariant phase-space observables in Ref.~\cite{carmeli2004position}. In the reconstructed representation, conjugation by a boost sends $\vect{P}$ to $\vect{P}+m\vect{v}$, equivalently the spectral measure of $\vect{P}$ obeys Eq.~\eqref{eq:momcov2}. Translations leave momentum unchanged, and covered rotations act via $\pi(u)$ on momentum space. The general covariant momentum observable is therefore a rotationally invariant smearing of the spectral measure of $\vect{P}$. Sharpness forces the smearing measure to collapse to a point, and rotation invariance fixes that point to the origin.
\end{proof}

\section{Orbital and spin parts of angular momentum}

For completeness we record the angular-momentum decomposition in a form independent of the induced-representation proof.

\begin{proposition}[Orbital plus spin]\label{prop:orbitalspin}
Suppose a self-adjoint triple $\vect{J}$ satisfies
\begin{equation}
[J_i,X_j]=\ii\hbar\epsilon_{ijk}X_k,
\qquad
[J_i,P_j]=\ii\hbar\epsilon_{ijk}P_k,
\end{equation}
in the representation of Theorem~\ref{thm:main}. Then
\begin{equation}
\vect{J}=\vect{X}\times\vect{P}+\vect{S},
\end{equation}
where $\vect{S}$ commutes with all $X_i$ and $P_i$ and acts only on the finite-dimensional spin factor.
\end{proposition}

\begin{proof}
Define the orbital part $L_i:=\epsilon_{ijk}X_jP_k$ and set $S_i:=J_i-L_i$. The canonical commutation relations imply
\begin{equation}
[L_i,X_j]=\ii\hbar\epsilon_{ijk}X_k,
\qquad
[L_i,P_j]=\ii\hbar\epsilon_{ijk}P_k.
\end{equation}
Subtracting from the assumed commutators for $J_i$ gives
\begin{equation}
[S_i,X_j]=0,
\qquad
[S_i,P_j]=0.
\end{equation}
By the Stone-von Neumann type structure already fixed by Mackey imprimitivity, the commutant of the position-momentum algebra on $L^2(\R^3)\otimes\C^{2s+1}$ is exactly the algebra acting on the internal factor. Hence $\vect{S}$ acts only on $\C^{2s+1}$.
\end{proof}

\section{Logical dependencies}

For reference, the argument uses the following chain of inputs and consequences.

\begin{enumerate}
\item A smooth family of reference states around an isotropic equilibrium, together with a smooth local completion of the Giannelli-Chiribella single-state observable assignment.
\item From these, the directional observable-generator duality on the tangent space of the reference-state manifold.
\item A covariant localization POVM together with fiducial focusing at the origin and the standard smearing structure for covariant localization observables.
\item From these, the global norm-one property of localization, sharp localization, and hence the canonical position PVM.
\item Local inertial frame composition, including time-rotation compatibility, and a central boost-translation holonomy.
\item From these, the representation $L^2(\R^3)\otimes\C^{2s+1}$, canonical momentum, scalar mass, the boost generator $m\vect{X}$, the free Hamiltonian $\vect{P}^{\,2}/2m+E_0$, the relation $\vect{P}=m\dot{\vect{X}}$, and the orbital-spin decomposition of angular momentum.
\end{enumerate}

The paper therefore isolates a route to the nonrelativistic one-particle sector without claiming a derivation of infinite-dimensional quantum theory from the finite-dimensional axioms of Ref.~\cite{giannelli2026information} alone.

\bibliography{refs}

\end{document}